\documentclass{article}
\usepackage{pagecolor, setspace, amsmath, amssymb, amsthm}
\usepackage[colorlinks = true]{hyperref}
\usepackage[margin = 2.5cm]{geometry}

\usepackage[english]{babel}
\usepackage[utf8x]{inputenc}
\usepackage{amsmath}
\usepackage{graphicx}
\usepackage[usenames,dvipsnames]{pstricks}
\usepackage{epsfig}
\usepackage{pst-grad} 
\usepackage{pst-plot} 
\usepackage{epstopdf}
\usepackage{subcaption}
\usepackage{tikz}
\usetikzlibrary{calc}
\usetikzlibrary{fit}
\newtheorem{theorem}{Theorem}[section]
\newtheorem{lemma}[theorem]{Lemma}

\newtheorem{corollary}[theorem]{Corollary}
\newtheorem{procedure}{Procedure}
\newtheorem{definition}{Definition}
\newtheorem{example}{Example}

\title{On Characteristic and Permanent Polynomials of a Matrix}
\author{Ranveer Singh \thanks{Center for System Science, Indian Institute of Technology, Jodhpur. email: \texttt{pg201283008@iitj.ac.in}}\\ R. B. Bapat \thanks{Stat-Math Unit, Indian Statistical Institute, Delhi, 7-SJSS Marg, New Delhi - 110 016. email: \texttt{rbb@isid.ac.in}}}

\begin{document}
        \maketitle
%

\begin{abstract}
There is a digraph corresponding to every square matrix over $\mathbb{C}$. We generate a recurrence relation using the Laplace expansion to calculate the characteristic, and permanent polynomials of a square matrix. Solving this recurrence relation, we found that the characteristic, and permanent polynomials can be calculated in terms of characteristic, and permanent polynomials of some specific induced subdigraphs of blocks in the digraph, respectively. Interestingly, these induced subdigraphs are vertex-disjoint and they partition the digraph. Similar to the characteristic, and permanent polynomials; the determinant, and permanent can also be calculated.  Therefore, this article provides a combinatorial meaning of these useful quantities of the matrix theory.  We conclude this article with a number of open problems which may be attempted for further research in this direction.
\end{abstract}

\emph{Keywords.} Determinant; Permanent; Block; Block Graph.\\
    \emph{AMS Subject Classifications.} 15A15, 05C20, 68R10.
    
%


%


\section{Introduction}
The problem of finding determinant, and permanent of  a matrix has been well studied in literature \cite{abdollahi2012determinants,helton2009determinant,bibak2013determinant,pragel2012determinants,huang2012determinant,bibak2013determinants,bapat2014adjacency,hwang2003permanents,harary1969determinants,farrell2000permanents,wanless2005permanents,minc1984permanents}. Methods have also been developed to calculate determinant of a matrix by utilizing digraph representations of the corresponding matrix \cite{harary1962determinant,greenman1976graphs}. The well-known problem, proposed by Collatz and Sinogowitz in 1957, is to characterize graphs with positive nullity \cite{von1957spektren,bibak2013determinant}. The zero determinant of the adjacency matrix of a graph ensures its positive nullity. Nullity of graphs is applicable in various branches of science, in particular, quantum chemistry, Hückel molecular orbital theory \cite{lee1994chemical,gutman2011nullity}, and social network theory \cite{leskovec2010signed}. The permanent of a square matrix has significant graph theoretic interpretations. It is equivalent to find out the number of cycle-covers in the directed graph corresponding to its adjacency matrix. Also the permanent is equivalent to the number of the perfect matching in the bipartite graph corresponding to its biadjacency matrix. Theory of permanents provides an effective tool in dealing with order statistics corresponding to random variables which are independent but possibly nonidentically distributed  \cite{bapat1989order}. Computing permanent of a matrix is a ``\#P-hard problem" which can not be done in polynomial time unless $P^{\# P}=P$, and in particular, $P=NP$ \cite{valiant1979complexity,wei2010matrix}. The characteristic polynomial of a square matrix $A$, of order $n$, is given by determinant of matrix $(A-\lambda I)$, where $I$ is an identity matrix of order $n$. We denote characteristic polynomial, $\det(A-\lambda I)$ by $\phi(A)$. Similarly,  the permanent polynomial of $A$ is given by permanent of matrix $(A-\lambda I)$. We denote permanent polynomial, per$(A-\lambda I)$ by $\psi(A)$.

    Now, we give some preliminaries and notations that are used in this article. We begin with a standard theorem on the Laplace expansion of determinant and permanent of a square matrix $A$. 
    
    \begin{theorem} \label{Godsil}
    Let $A_{S,T}$ denote the submatrix of matrix $A$ of order $n$ with rows indexed by elements in $S$, and columns by the elements in $T$. Let $A'_{S,T}$ denote the submatrix of $A$ with rows indexed by elements not in $S$, and columns by the elements not in $T$,$$\det (A)=\sum_{T}(-1)^{w(S,T)}\det (A_{S,T})\det (A'_{S,T}).$$ $$\emph{per}(A)=\sum_{T}\emph{per} (A_{S,T})\emph{per} (A'_{S,T}).$$  Here, $S$ is a fixed $k$-subset of the rows of $A$;  $T$ runs over all $k$-subsets of the columns of $A$, for $k<n$. Also, $w(S,T)=|S|+|T|$, where $|S|$, and $|T|$ represents sum of all the elements in $S$, and $T$, respectively. 
    \end{theorem}
    
    A digraph $G = \Big(V(G), E(G)\Big)$ is a collection of a vertex set $V(G)$, and an edge set $E(G)\subseteq V(G) \times V(G)$. An edge $(u,u)$ is called a loop at the vertex $u$. A simple graph is a special case of a digraph, where $E(G) = \{(u,v): u \neq v \}$; and if $(u,v) \in E(G)$, then $(v,u) \in E(G)$. A weighted digraph is a digraph  equipped with  a weight function $f: E(G) \rightarrow \mathbb{C}$.
    If $V(G) = \emptyset$ then, the digraph $G$ is called a null graph. A subdigraph of $G$ is a digraph $H$, such that, $V(H) \subseteq V(G)$ and $E(H) \subseteq E(G)$. The subdigraph $H$ is an induced subdigraph of $G$ if $u, v \in V(H)$ and $(u,v) \in E(G)$ indicate $(u,v) \in E(H)$. Two subdigraphs $H_1$, and $H_2$ are called vertex-disjoint subdigraphs if $V(H_1)\cap V(H_2)=\emptyset$. A path of length $k$ between two vertices $v_1$, and $v_k$ is a sequence of distinct vertices $v_{1}, v_{2}, \hdots, v_{k-1}, v_{k}$, such that, for all $i=1,2, \hdots, k-1$,  either $(v_{i}, v_{i +1}) \in E(G)$ or $(v_{i + 1}, v_i) \in E(G)$. We call a digraph $G$ be connected, if there exist a path between any two distinct vertices. A component of $G$ is a maximally connected subdigraph of $G$. A cut-vertex of $G$ is a vertex whose removal results increase the number of components in $G$. Now, we like to define the idea of a block of a digraph, which plays a fundamental role in this article. It is already defined in literature for simple graphs \cite{bapat2014adjacency}. 
    
    \begin{definition}{\bf Block:} 
    A block is a maximally connected subdigraph of $G$ that has no cut-vertex.
    \end{definition}
    Note that, if $G$ is a connected digraph having no cut-vertex, then $G$ itself is a block. A block is called a pendant block if it contains only one cut-vertex of $G$, or it is the only block in that component. The blocks in a digraph can be found in linear time using John and Tarjan algorithm \cite{hopcroft1971efficient}. We define the cut-index of a cut-vertex $v$, with the number of blocks connected to $v$. We specifically denote a digraph having $k$ blocks with $G_k$.
    
    A square matrix $A=(a_{uv})\in \mathbb{C}^{n\times n}$ can be depicted by a weighted digraph $G(A)$ with $n$ vertices. If $a_{uv} \neq 0$, then $(u,v) \in E\Big(G(A)\Big)$, and $f(u,v) = a_{uv}$. The diagonal entry $a_{uu}$ corresponds to a loop at vertex $u$ having weight $a_{uu}$. If $v$ is a cut-vertex in $G(A)$, then we call $a_{vv}$ as the corresponding cut-entry in $A$. The following example will make this assertion transparent.

    \begin{example}
    The digraphs corresponding to the matrices $M_1$, and $M_2$ are presented in figure \ref{fig1}.
    $$M_1= \begin{bmatrix}
    0& 3& 2& 0& 0& 0& 0\\
    -7 &\color{red}5& -1& 1& -8& 0& 0\\
    2 &-1& 0& 0& 0& 0& 0\\
    0& 1& 0& 0 & 0& -3& 0\\
    0 & 12& 0& 0& 0& 1& 0\\
    0 &0& 0& 1& 1& \color{red}-4& 2\\
    0& 0& 0& 0& 0& 20& 3
    \end{bmatrix}, \ \ \ \ M_2= \begin{bmatrix}
    0& 3& 2& 0& 0& 0& 0& 0\\
    -7 &\color{red}5& -1& 1& -8& 0& 0&0\\
    2 &-1& 0& 0& 0& 0& 0&0\\
    0& 1& 0& 0 & 0& -3& 0&0\\
    0 & 12& 0& 0& 0& 1& 0&0\\
    0 &0& 0& 1& 1& \color{red}-4& 2&-2\\
    0& 0& 0& 0& 0& 20& 3&0\\
    0& 0& 0& 0& 0& -2& 0&10
    \end{bmatrix}.$$
    
    \begin{figure}
    \centering
    \begin{subfigure}{.5\textwidth}
      \centering
      \includegraphics[width=0.7\linewidth]{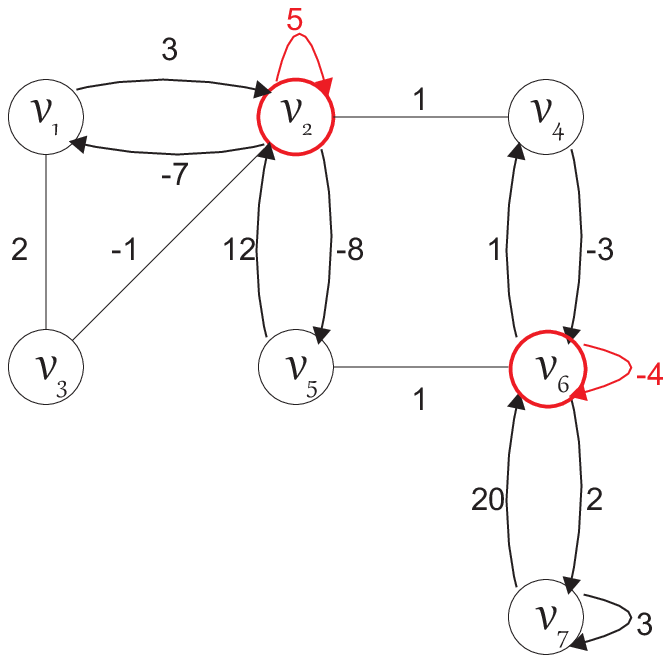}
      \caption{}
      \label{fig:sub1}
    \end{subfigure}%
    \begin{subfigure}{.5\textwidth}
      \centering
      \includegraphics[width=0.7\linewidth]{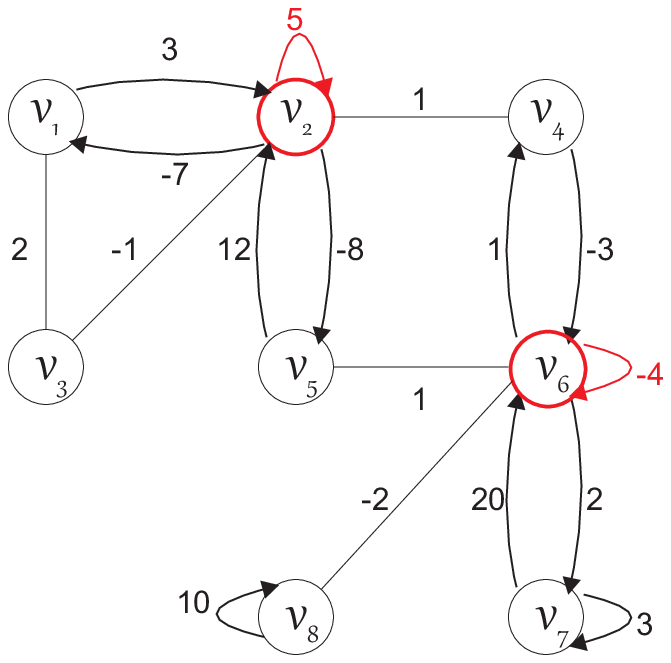}
      \caption{}
      \label{fig:sub2}
    \end{subfigure}
    \caption{(a) Digraph of matrix $M_1$ \ (b) Digraph of matrix $M_2$}
    \label{fig1}
    \end{figure}
    
    The cut-entries, and the cut-vertices are shown in red in the matrices $M_1$, and $M_2$, as well as in their corresponding digraphs $G(M_1)$, and $G(M_2)$. Note that, when $a_{uv} = a_{vu} \neq 0$, we simply denote edges $(u,v)$, and $(v,u)$ with an undirected edge $(u,v)$ with weight $a_{uv}$. As an example, in $G(M_1)$ the edge $(v_1, v_3)$ and $(v_3,v_1)$ are undirected edges.   The digraph $G(M_1)$, depicted in the figure \ref{fig1}(a), has blocks $B_1, B_2$, and $B_3$ which are induced subdigraphs on vertex subsets $\left\{v_1,v_2,v_3\right\}, \left\{v_2,v_4,v_5,v_6\right\}$, and $\left\{v_6,v_7\right\}$, respectively. Here, cut-vertices are $v_2$, and $v_6$ with cut-indices 2. Similarly, the digraph $G(M_2)$ in the figure \ref{fig1}(b), has blocks $B_1, B_2, B_3$, and $B_4$ on vertex sets $\left\{v_1,v_2,v_3\right\}, \left\{v_2,v_4,v_5,v_6\right\},\left\{v_6,v_7\right\}$, and $\left\{v_6,v_8\right\}$, respectively. Here, cut-vertices are $v_2$, and $v_6$ with cut-indices $2$, and $3$, respectively.
    \end{example}

    The characteristic polynomial of a digraph $G(A)$ is the characteristic polynomial of a matrix $A$. In other words, $\phi\Big(G(A)\Big) = \phi(A) = \det(A-\lambda I)$. Hence, the determinant of digraph $G(A)$ is the determinant of a matrix $A$. Similarly, the permanent polynomial of digraph $G(A)$ is the permanent polynomial of corresponding matrix $A$, or in other words $\psi\Big(G(A)\Big)=\psi(A)=$per$(A-\lambda I)$. Hence, the permanent of digraph $G(A)$ is the permanent of a matrix $A$.

    Inspired by a work on simple block graphs \cite{bapat2014adjacency}, we propose a new technique for computing the characteristic, and permanent polynomials of a matrix. First of all, we derive a recursive expression for these polynomials of a matrix with respect to a pendant block in the corresponding digraph. On solving this recursive expression we found that characteristic(permanent) of a digraph can be written in terms of the characteristic (permanent) polynomial of some specific induced subdigraphs of blocks. Interestingly, these induced subdigraphs are vertex-disjoint and they partition the digraph. Hence, this leads us to define a new partition called $\mathcal{B}$-partition of a digraph. Corresponding to every $\mathcal{B}$-partition we define the $\phi$-summand, and $\psi$-summand. Similarly, the $\det$-summand, and per-summand corresponding to each $\mathcal{B}$-partition is specified. Thus, we have found the characteristic, and permanent polynomials of a matrix in terms of $\phi$-summands, and $\psi$-summands, respectively, of the corresponding $\mathcal{B}$-partitions. Similarly,  the determinant, and permanent of the matrix can be found in terms of $\det$-summands, and per-summands, respectively. 
    This new method of calculation provides a combinatorial significance of the determinant, permanent, characteristic, and permanent polynomials of a matrix. A singular graph has a zero eigenvalue. Classifying singular graphs is a complicated problem in combinatorics \cite{sciriha2007characterization, bapat2011note, bapat2014adjacency}. In this article, we illuminate this problem with a number of examples with the new combinatoric implication. This procedure presents a simplified proof for the determinant of simple block graphs earlier given in \cite{bapat2014adjacency}. These graph-theoretic representations would be useful in future investigations in matrix theory.
    
    The paper is organized as follows. In section \ref{cpdg}, we derived a recursive form of the characteristic, and permanent polynomials of the digraph. In section \ref{betapartition}, we define $\mathcal{B}$-partition of digraph, and its corresponding $\phi$-summand, $\psi$-summand, $\det$-summand, and per-summand.     
    In section \ref{solutionofrecur}, we solve recursive expression derived in section \ref{cpdg} to get the characteristic, and permanent polynomials of digraph in terms of $\phi$-summands, and $\psi$-summands, respectively. In section \ref{detdi}, we provide results on the determinant of some simple graphs including block graph. In section \ref{future}, we discuss the case when digraph has no cut-vertex.  Finally, in section \ref{con}, we conclude the article and put some open problems.

    \section{Recursive form of characteristic and permanent polynomial} \label{cpdg}
    In this section, we provide a recursive expression for the characteristic and permanent polynomials of a digraph $G$, with respect to its pendant block. Let $Q$ be a subdigraph of $G$, then $G\setminus Q$ denotes the induced subdigraph of $G$ on the vertex subset $V(G)\setminus V(Q)$. Here, $V(G)\setminus V(Q)$ is the standard set-theoretic subtraction of vertex sets. Let $v$ be a cut-vertex of $G$, then the recursive expression for these polynomials, can also be given with respect to subdigraph $H$ containing $v$, such that, $H\setminus v$ is a union of components. For convenience, we relabel the digraph $G(A)$. In graph theory, these relabelling are captured by permutation similarity of $A$. Determinant of permutation matrices is equal to 1. Thus, relabelling on vertex set keep the determinant and the permanent unchanged. We frequently use this idea in this section.
    
    \begin{lemma} \label{reccurcp}
    Let $G$ be a digraph, having $B_1$ as a pendant block and $v$ be the cut-vertex of $G$ in $B_1$. Let the weight of loop at vertex $v$ be $\alpha$. The following recurrence relation holds for the characteristic polynomial,   
    $$\phi(G)=\phi(B_1)\times\phi(G\setminus B_1)+\phi(B_1\setminus v)\times \phi\Big(G\setminus(B_1\setminus v)\Big)+(\lambda-\alpha)\times\phi(B_1\setminus v)\times\phi(G\setminus B_1).$$     
    \end{lemma}   
    \begin{proof}
    Let $A(G)$ be the matrix corresponding to the digraph $G$ having $n$ vertices. Let the number of vertices in block $B_1$ is $n_1$. With suitable reordering of vertices in $G$, let block $B_1$ has vertices with the labels \{1,2,$\hdots$,$n_1$\}. Here,  $n_1$-th vertex be the cut-vertex $v$ of $G$ in $B_1$. Let $x$,and $z$ are column vectors of order $(n_1-1)$, and $(n- n_1)$, respectively, such that $(x, \alpha, z)$ formulates $n_1$-th column vector of $A(G)$. Similarly, $w$, and $y$ are row vectors of order $(n_1-1)$, and $(n- n_1)$, respectively, such that $(w, \alpha, y)$ formulates $n_1$-th row vector of $A(G)$. Now,
    \begin{equation}
    A(G)-\lambda I=\begin{bmatrix}
    A(B_1\setminus v)-\lambda I& x & O \\
    w& \alpha-\lambda & y \\
    O & z& A(G\setminus B_1)-\lambda I \end{bmatrix}.
    \end{equation}

    Here, $O$ is the zero matrix of appropriate size. Let $x_i, w_i, z_i$, and $y_i$ denote the $i$-th entry of the vectors $x, w, z$, and $y$, respectively. Using theorem \ref{Godsil}, let us fix set $S = \{1,2,\hdots,n_1\}$, then $T$ is $n_1$-subset of $n$ columns in $A(G)-\lambda I$. Note that, $T$ must have numbers  $1,2,3\hdots,(n_1-1)$ otherwise it will give zero contribution to $\phi(G)$. This is because if any number $0\le i \le (n_1-1)$ is missing in $T$, then column corresponding to $i$ will be in submatrix $[A(G)-\lambda I]'_{S,T}$. As this column in submatrix $[A(G)-\lambda I]'_{S,T}$ has all zero entries, hence its determinant is zero.  Thus, only following sets of $S, T$ have possible non-zero contribution in $\phi(G)$.

    \begin{enumerate}
    \item $S=\{1,2,\hdots,n_1\}, \ T=\{1,2,\hdots,n_1\}.$ It contributes the following to $\phi(G)$
    \begin{equation}
    \phi(B_1)\times\phi(G\setminus B_1)
    \end{equation}

    \item $S=\{1,2,\hdots, n_1\}, \ T=\{1,2,\hdots, n_1-1,r \}$, where $r$ is a value from set $\{n_1+1,n_1+2,\hdots,n\}$. So there will be $n-n_1$ such possible set of $T$. Let  $T_i=\{1,2,\hdots,n_1-1,n_1+i\}$  for $i=1,2,\hdots,n-n_1.$
    
    Let $c_i$ be the $i$-th column of $A(G\setminus B_1)-\lambda I$, then contribution of these sets to $\phi(G)$ is
    \begin{equation} 
    \begin{split} 
    & \sum_{i=1}^{n-n_1}(-1)^{w(S,T_i)}\det[A(G)-\lambda I]_{S,T}\times\det[A(G)-\lambda I]'_{S,T}\\
    = & \sum_{i=1}^{n-n_1}(-1)^{w(S,T_i)}\det\begin{bmatrix}
    A(B_1\setminus v)-\lambda I&O\\
    w&  y_i
    \end{bmatrix}\times\det \begin{bmatrix}
    z& \Big(A(G\setminus B_1)-\lambda I\Big)\setminus c_i
    \end{bmatrix}\\
    = & \det(B_1\setminus v-\lambda I)\times\Bigg(\sum_{i=1}^{n-n_1}(-1)^{w(S,T_i)}\times y_i\times \det \begin{bmatrix}
    z& \Big(A(G\setminus B_1)-\lambda I\Big)\setminus c_i
    \end{bmatrix}\Bigg)\\
    = & \det(B_1\setminus v-\lambda I)\times\Bigg(\det\Big(G\setminus(B_1\setminus v)-\lambda I\Big)+(\lambda-\alpha)\times\det\Big(A(G\setminus B_1)-\lambda I\Big)\Bigg)\\
    = & \phi(B_1\setminus v)\times\Bigg(\phi\Big(G\setminus(B_1\setminus v)\Big)+(\lambda-\alpha)\phi(G\setminus B_1)\Bigg).
    \end{split}
    \end{equation} 
    Hence,
    
    \begin{equation}
    \phi(G)=\phi(B_1)\times\phi(G\setminus B_1)+\phi(B_1\setminus v)\times \phi\Big(G\setminus(B_1\setminus v)\Big)+(\lambda-\alpha)\times\phi(B_1\setminus v)\times\phi(G\setminus B_1)
    \end{equation}   
    \end{enumerate} 
    \end{proof}

    \begin{corollary} \label{reccurpp}
    Let $G$ be a digraph, having $B_1$ as a pendant block and $v$ be the cut-vertex of $G$ in $B_1$. Let the weight of loop at vertex $v$ be $\alpha$. The following recurrence relation holds for the permanent polynomial,  $$\psi(G)=\psi(B_1)\times\psi(G\setminus B_1)+\psi(B_1\setminus v)\times \psi\Big(G\setminus(B_1\setminus v)\Big)+(\lambda-\alpha)\times\psi(B_1\setminus v)\times\psi(G\setminus B_1)$$   
    \end{corollary}
    \begin{proof}
    The proof is similar to lemma \ref{reccurcp}.
    \end{proof}

    We generalize Lemma (\ref{reccurcp}) with respect to some subdigraphs containing cut-vertex in the next lemma. 
    
    \begin{lemma} \label{recurrsubp}
    Let $G$ be a digraph with at least one cut-vertex. Let $H$ be a non empty subdigraph of $G$ having cut-vertex $v$ with loop weight $\alpha$, such that $H\setminus v$ is union of components. The characteristic polynomial of $G$,
    $$\phi(G)=\phi(H)\times\phi(G\setminus H)+\phi(H\setminus v)\times \phi\Big(G\setminus(H\setminus v)\Big)+(\lambda-\alpha)\times\phi(H\setminus v)\times\phi(G\setminus H).$$ 
    \end{lemma}
    \begin{proof}
    Let the number of vertices in subdigraph $H$ is $n_1.$ With suitable reordering of vertices in $G$,  let $n_1$-th vertex be the cut-vertex $v$ of $G$ in $H$. Remaining  proof is similar to lemma(\ref{reccurcp}).
    \end{proof}
    
    \begin{corollary}
    Let $G$ be a digraph with at least one cut-vertex. Let $H$ be a non empty subdigraph of $G$ having cut-vertex $v$ with loop weight $\alpha$, such that $H\setminus v$ is union of components. The permanent polynomial of $G$,
    $$\psi(G)=\psi(H)\times\psi(G\setminus H)+\psi(H\setminus v)\times \psi\Big(G\setminus(H\setminus v)\Big)+(\lambda-\alpha)\times\psi(H\setminus v)\times\psi(G\setminus H).$$ 
    \end{corollary}
    \begin{proof}
    The proof follows from lemma \ref{recurrsubp}.
    \end{proof}

    \section{$\mathcal{B}$-partitions of a digraph} \label{betapartition}
    We define a new partition of digraph which helps in finding its characteristic (permanent) polynomial.

    \begin{definition} \label{def1}
     Let $G_k$ be a digraph having $k$ blocks $B_1, B_2, \hdots B_k$. Then, a $\mathcal{B}$-partition of $G_k$ is a partition in $k$ vertex disjoint induced subdigraphs $\hat{B_1}, \hat{B_2}, \hdots, \hat{B_k}$, such that, $\hat{B_i}$ is a subdigraph of $B_i$. The $\phi$-summand, and $\det$-summand of this $\mathcal{B}$-partition is  $$\prod_{i}^{k}\phi (\hat{B_i}), ~\text{and}~ \prod_{i}^{k}\det (\hat{B_i}),$$ respectively, where by convention $\phi(\hat{B_i})=1, ~\text{and}~\ \det(\hat{B_i})=1 $ if $\hat{B_i}$ is a null graph. Similarly, its $\psi$-summand, and $\emph{per}$-summand is given by $$\prod_{i}^{k}\psi  (\hat{B_i}), ~\text{and}~ \ \ \prod_{i}^{k}\emph{per} (\hat{B_i}),$$ respectively, where by convention $\psi (\hat{B_i})=1, ~\text{and}~ \emph{per} (\hat{B_i})=1$ if $\hat{B_i}$ is a null graph.
    \end{definition}

    Examples of a $\mathcal{B}$-partition of digraph of matrix $M_1,$ and $M_2$ are shown in figure \ref{fig2}(a),\ (b), respectively. All possible $\mathcal{B}$-partitions of $M_1,$ and $M_2$ are given in appendix in figure: \ref{allbm1}, and \ref{allbm2}, respectively.
    \begin{figure}
    \centering
    \begin{subfigure}{.5\textwidth}
      \centering
      \includegraphics[width=0.7\linewidth]{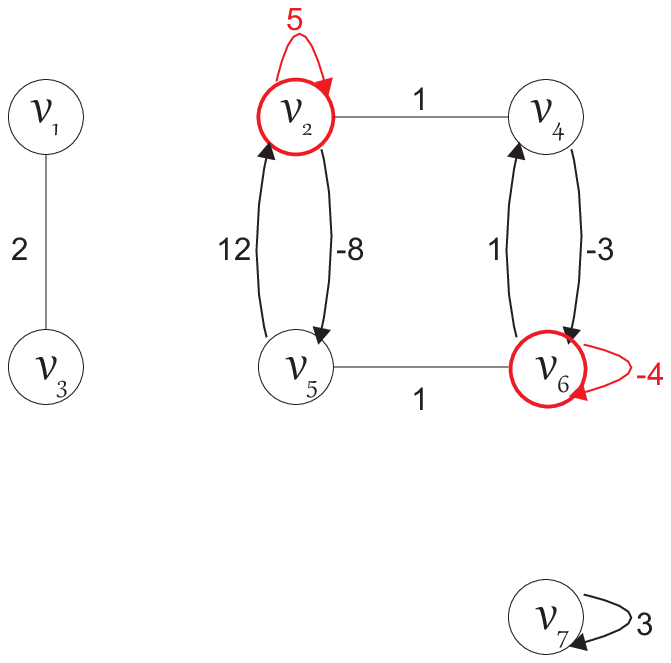}
      \caption{}
      \label{fig:sub1}
    \end{subfigure}%
    \begin{subfigure}{.5\textwidth}
      \centering
      \includegraphics[width=0.7\linewidth]{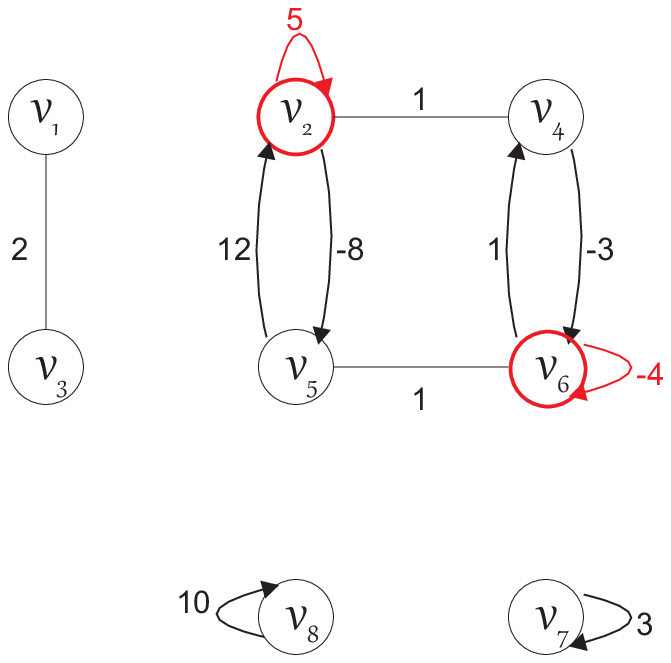}
      \caption{}
      \label{fig:sub2}
    \end{subfigure}
    \caption{Example of a $\mathcal{B}$-partition of (a) Digraph of matrix $M_1$ \ (b) Digraph of matrix $M_2$}
    \label{fig2}
    \end{figure}
    
    \begin{corollary}
    Let $G$ be a digraph having $t$ cut-vertices with cut-indices $d_1,d_2,\hdots,d_t$, respectively. The number of $\mathcal{B}$-partitions of $G$ is  $$\prod_{i=1}^{t}d_i.$$  
    \end{corollary}
    \begin{proof}
    Each cut-vertex associates with an induced subdigraph of exactly one block in a $\mathcal{B}$-partition. For $i$-th cut-vertex there are $d_i$  choices of blocks. Hence, the result follows. 
    \end{proof}
    
    The number of $\mathcal{B}$-partitions in digraphs corresponding to matrices $M_1$, and $M_2$ are $4$, and $6$,  respectively.

    \section{Characteristic and Permanent Polynomial of Matrix} \label{solutionofrecur}
    In this section, we derive an expression for characteristic polynomial of an arbitrary square matrix using $\phi$-summands in $\mathcal{B}$-partitions of the digraph corresponding to the matrix. Similarly, we apply $\psi$-summands for the permanent polynomial.
    
    Let $G_k$ be a digraph having $k$ blocks $B_1, B_2, \hdots, B_k$. Let $G_k$ has $m$  cut-vertices with cut-indices $d_1, d_2, \hdots, d_m$. Also,  weights of loops at these vertices are $\alpha_1, \alpha_2, \hdots, \alpha_m$, respectively. Then, the followings are steps to calculate characteristic, and permanent polynomials of $G_k$. \\
    \begin{procedure}
    $\vspace{2mm}$
    \begin{enumerate}
    \item Add $-\lambda$ with the loop-weight at each vertex. Whenever, there is no loop at a vertex then, add a loop with weight $-\lambda$.
    \item For $q=0,1,2,\dots,m,$\begin{enumerate}
    \item Delete any $q$ cut-vertices at a time from $G_k$ to construct an induced subdigraph. In this way, construct all $\binom{m}{q}$ induced subdigraphs of $G_k$.
    \item Find all possible $\mathcal{B}$-partitions of each subdigraphs constructed in (a).
    
    \item For each $\mathcal{B}$-partition in (b)  multiply its $\phi$-summand($\psi$-summand) by $\prod_{i}^{}(\lambda-\alpha_i)(d_i-1).$ Here, $i=1,2,\dots,q$. Also, $\alpha_i$, and $d_i$ are weight and cut-index of removed $i$-th cut-vertex, respectively.  
    \end{enumerate}
     
    \item Sum all the terms in 2(c).
    
    \end{enumerate}
    
    \end{procedure}
    In the next theorem, we justify the above procedure to find characteristic polynomial of an arbitrary matrix.

    \begin{theorem} \label{matrixcp}
    Let $G_k$ be a digraph with $m$ cut-vertices, and $k$ blocks $B_1, B_2,\hdots, B_k$. Let $G^q_{k}=\left\{G^q_{k1}, G^q_{k2},\hdots, G^q_{k\binom{m}{q}}\right\}$ be set of all induced subdigraphs of $G_k$ after removing any $q$ cut-vertices. Also, let $d_{i1}, d_{i2},\hdots, d_{iq}$ be cut-indices, and $\alpha_{i1}-\lambda, \alpha_{i2}-\lambda,\hdots, \alpha_{iq}-\lambda$ be weights of loop of these removed $q$ cut-vertices to form $G^q_{ki}$. Let $S^q_{ki}$ denote summation of all $\phi$-summands of all possible $\mathcal{B}$-partition of $G^q_{ki}.$ Then, the characteristic polynomial of $G_k$ is given by,
    $$\phi(G_k)=\sum_{q=0}^{m}\varphi(G^q_{k}), ~\text{where}~ \varphi(G^q_{k})=\sum_{i=1}^{\binom{m}{q}}\Big(L^q_{i} S^q_{ki}\Big) ~\text{and}~ L^q_{i}=\prod_{t=1}^{q}(\lambda-\alpha_{it})(d_{it}-1),\ \ \ L^0_{1}=1.$$

    \end{theorem}
    \begin{proof}
    We use method of mathematical induction on $G_k$, with $k \in \mathbb{N}$ to prove the theorem.
    \begin{enumerate}
    \item
    Let $k = 1$. Given any digraph $G_1$ having only one block, 
    $$\phi(G_1)=\varphi(G^0_1)=L^0_1S^0_{11}=S^0_{11}.$$ 
    Here is only one $\mathcal{B}$-partition, which is block itself, so $\phi$-summand is $\phi(G_1)$, that is, $S^0_{11}=\phi(G_1)$. Hence, this theorem is valid in this case.\\  
    
    \item 
    For $k=2$, digraph $G_2$ has two blocks with one cut-vertex. Let they be $B_1$, and $B_2$. Also, let $v$ be the cut-vertex with cut-index 2, and the loop-weight at vertex $v$ be  $\alpha_1$. Thus,

    $$G^0_{2}=\left\{G^0_{21}\right\}, \ G^1_{2}=\left\{G^1_{21}\right\},$$ 
    where, $G^0_{21}=G_2,\ G^1_{21}=G_2\setminus v.$ Now we find $\varphi(G^0_2), \varphi(G^1_2)$. 
    \begin{enumerate}
    \item 
    Note that, to compute $\varphi(G^0_2)$ there is no cut-vertex to remove. Thus, $q=0$, $\binom{m}{q}=1$, and $L^0_{1}=1.$ There are two $\mathcal{B}$-partition of $G^0_{21}$. One of them consists of induced subdigraphs $B_1$, and $B_2\setminus v$. Another one contains the induced subdigraphs  $B_1\setminus v$, and $B_2$. Hence, 
    $$S^0_{21}=\phi(B_1)\phi(B_2\setminus v)+\phi(B_2)\phi(B_1\setminus v)=\varphi(G^0_2).$$
    
    \item 
    Also, to compute $\varphi(G^1_2)$ there is only cut-vertex to remove. Here, $q=1$, $\binom{m}{q}=1$, and $d_{11}=2$. Thus, $L^1_{1}=(\lambda-\alpha_1).$ The only possible $\mathcal{B}$-partition of $G^1_{21}$ contains the induced subdigraphs $B_1\setminus v$, and $B_2\setminus v$. Hence,
    \begin{equation}
    \begin{split}
        & S^1_{21}=\phi(B_1\setminus v)\phi(B_2\setminus v), \\
         \text{therefore,}~ & \varphi(G^1_{2})=(\lambda-\alpha_1) \phi(B_1\setminus v)\phi(B_2\setminus v).
    \end{split}
    \end{equation}

    \end{enumerate}
    On combining $(a)$, and $(b)$ we observe, 
    $$\phi(G_2)=\phi(B_1)\phi(B_2\setminus v)+\phi(B_2)\phi(B_1\setminus v)+(\lambda-\alpha_1)\phi(B_1\setminus v)\phi(B_2\setminus v).$$ 
    
    Using Lemma \ref{reccurcp} on $G_2$ we also obtain the above expression. It proves the theorem for $k = 2$. 
     
    Hence theorem is true for $G_2.$\\

    \item 
    Now, we assume the theorem is true for any $G_n$ for $2< n\le k$. We need to prove the theorem for $G_{n+1}$. We can always select a pendant block of $G_{n+1}$ and denote it as $B_{n+1}$.  Let $v$ be the cut-vertex of $G_{n+1}$ in $B_{n+1}$, having loop-weight $(\alpha-\lambda)$. Note that, the digraph $G_{n+1}$ has an induced subdigraph $G_{n}=\Big(G_{n+1}\setminus(B_{i+1}\setminus v)\Big)$. The theorem is true for $G_n$ by the assumption of induction. 
    Then, from Lemma \ref{reccurcp}, 
    \begin{equation}\label{G_{n+1}}
    \phi(G_{n+1})=\phi(G_n)\phi(B_{n+1}\setminus v)+\phi(G_n\setminus v)\phi(B_{n+1})+(\lambda-\alpha )\Big(\phi(G_n\setminus v)\phi(B_{n+1}\setminus v)\Big).
    \end{equation}

    In this context, there are two cases depending on whether $v$ is also a cut-vertex of induced subdigraph $G_n$ or not. 
    
    \begin{enumerate}
    \item 
    Let the vertex $v$ is not a cut-vertex of $G_n$. In this case, the number of cut-vertices in $G_{n+1}$ is 1 more than that of $G_n$. In $G_{n+1}$ all the cut-vertices have same cut-indices as is in $G_{n}$, except $v$. The cut-index of $v$ in $G_{n+1}$ is 2. 
    
    \item 
    Let the vertex $v$ is a cut-vertex of $G_n$. 
     In this case, the number of cut-vertices in $G_{n+1}$ are same as the number of cut-vertices in $G_n$. In $G_{n+1}$ all the cut-vertices have same cut-indices as is in $G_{n}$, except $v$. If the cut-vertex $v$ has cut-index equals to $d_v$ in $G_{n}$, then it has cut-index $d_v+1$ in $G_{n+1}$. 
    %

    Now, we check whether theorem for $G_{n+1}$ is equivalent to equation (\ref{G_{n+1}}) for the both the cases. For each set of $q$  deleted cut-vertices from $G_{n+1}$ as required by theorem, there are two cases:  
    
    \begin{enumerate}
    \item The cut-vertex $v$ is not in deleted cut-vertices:
    
    Note that, in this case, cut-indices, and loops-weights of these removed $q$ cut-vertices in $G_{n+1}$ are same as those were in $G_n$. Thus, $L^q_i$ value corresponding to $G^q_{(n+1)i}$ remain same as for $G^q_{ni}$. Also, in this case, in all $\phi$-summands of $G_{n+1}$, of either  $\phi(B_{n+1})$ or $\phi(B_{n+1}\setminus v)$ has to be there. The first term in right hand side of equation (\ref{G_{n+1}}), that is,    $\phi(G_n)\phi(B_{n+1}\setminus v)$ gives all the required $\phi$-summands where $\phi(B_{n+1}\setminus v)$ is there. The second term in right hand side, that is,  $\phi(G_n\setminus v)\phi(B_{n+1})$ gives all required $\phi$-summands where $\phi(B_{n+1})$ is there. So, all the required $\phi$-summands of $G_{n+1}$ are generated when $q$ cut-vertices are deleted from $G_{n+1}$ and cut-vertex $v$ is not included in these $q$ cut-vertices. \\

    \item The cut-vertex $v$ is in deleted cut-vertices.\\
    As $v$ get removed, in this case, in $G_{n+1}$ all the $\phi$-summands must have term, $\phi(B_{n+1}\setminus v)$. The first term in right hand side of equation (\ref{G_{n+1}})  that is, $\phi(G_n)\phi(B_{n+1}\setminus v)$ gives  $\phi$-summands where $\phi(B_{n+1}\setminus v)$ is there. Note that, in $G_{n+1}$, in case of (a), cut-vertex $v$ has cut-index equal to 2, in case of (b), cut-vertex $v$ has cut-index equal to $d_v+1$. Hence, in both the cases there must be a extra $\phi$-summand multiplied by $(\lambda-\alpha)\phi(B_{n+1}\setminus v)$ corresponding to each $\phi$-summand in $G_n$ when these $q$ cut-vertices being removed. The third term in right hand side gives these extra summands. Therefore, all the $\phi$-partitions of $G_{n+1}$ can be obtained when $q$ cut-vertices are deleted from $G_{n+1}$ and $v$ is included in these $q$ cut-vertices. 
    \end{enumerate} 
    
    \end{enumerate} 
    \end{enumerate}
    Hence, the statement is true for $G_{n+1}$. This proves the theorem.

    \end{proof}
    
    \begin{corollary} \label{matrixpp}
    Let $G_k$ be a digraph with $m$ cut-vertices, and $k$ blocks $B_1, B_2,\hdots, B_k$. Let $G^q_{k}=\left\{G^q_{k1}, G^q_{k2},\hdots, G^q_{k\binom{m}{q}}\right\}$ be set of all induced subdigraphs of $G_k$ after removing any $q$ cut-vertices together. Also, let $d_{i1}, d_{i2},\hdots, d_{iq}$ be cut-indices, and $\alpha_{i1}-\lambda, \alpha_{i2}-\lambda,\hdots, \alpha_{iq}-\lambda$ be weights of loop of these removed $q$ cut-vertices to form $G^q_{ki}$. Let $S^q_{ki}$ denote summation of all $\psi$-summands of all possible $\mathcal{B}$-partition of $G^q_{ki}.$ Then, the permanent polynomial of $G_k$ is given by,
    $$\psi(G_k)=\sum_{q=0}^{m}\varphi(G^q_{k}), ~\text{where}~ \varphi(G^q_{k})=\sum_{i=1}^{\binom{m}{q}}\Big(L^q_{i} S^q_{ki}\Big) ~\text{and}~ L^q_{i}=\prod_{t=1}^{q}(\lambda-\alpha_{it})(d_{it}-1),\ \ \ L^0_{1}=1.$$ 
    
    \end{corollary}
    \begin{proof}
    Proof is similar to Theorem \ref{matrixcp}. 
    \end{proof}

    \begin{corollary}
    Let $G_k$ be a digraph having only one cut-vertex $v$ and $k$ blocks $B_1, B_2,\hdots, B_k$. Let $v$ has a loop of weight $\alpha$. Then,
    $$\phi(G_k)=\sum_{i=1}^{k}\Big(\phi(B_i)\prod_{j=1, j\ne i}^{k}\phi(B_j\setminus v)\Big)+(k-1)(\lambda-\alpha)\prod_{i=1}^{k}\phi(B_i\setminus v).$$
    \end{corollary}
    \begin{proof}
    The proof follows from theorem \ref{matrixcp}. On the right hand side, first term corresponds to $\phi$-summands of $\mathcal{B}$-partition when cut-vertex $v$ is not removed. The second term corresponds to $\phi$-summands of $\mathcal{B}$-partition when cut-vertex $v$ is removed. Cut-index of $v$ is $k$. Hence, $k-1$ is multiplied to the weight $(\lambda-\alpha)$.  
    \end{proof}

    \begin{example}
    We have constructed the digraphs of matrices $M_1$ and $M_2$ in figure \ref{fig1}. Now, we calculate the characteristic polynomial of these matrices in terms of the characteristic polynomial of the induced subdigraphs in the blocks in it, applying Theorem \ref{matrixcp}. We express the characteristic polynomial in terms of $\phi$-summands.  Let $X$ be a set of indexes. The principal submatrix whose rows and column are indexed with elements in $X$ is denoted by $[X]$. 
    
    \begin{enumerate}
    \item 
    First we calculate the characteristic polynomial of $M_1$. Parts of the characteristic polynomial is listed below in terms of $\phi$-summands of the digraph $G(M_1)$.
    \begin{enumerate}
    \item 
    Without removing any cut-vertex we get the following part of $\phi(M_1)$.
    \begin{equation*}
    \begin{split}
    \phi[1,2,3]\phi[4,5,6]\phi[7]+\phi[1,2,3]\phi[4,5]\phi[6,7]+\phi[1,3]\phi[2,4,5,6]\phi[7]\\+\phi[1,3]\phi[2,4,5]\phi[6,7].
    \end{split}
    \end{equation*}
    
    \item 
    Recall that, the loop-weight of $v_2$ is $(5 -\lambda)$ and its cut-index is 2. Removing the cut-vertex $v_2$ we get the following part, 
    \begin{equation}
        (\lambda-5) \big( \phi[1,3]\phi[4,5,6]\phi[7] + \phi[1,3]\phi[4,5]\phi[6,7] \big).
    \end{equation}
    
    \item the loop-weight of $v_6$ is $(-4-\lambda)$ and its cut-index is 2. Removing the cut-vertex $v_6$ we get the following part, \begin{equation}
    (\lambda+4)\big(\phi[1,2,3]\phi[4,5]\phi[7]+\phi[1,3]\phi[2,4,5]\phi[7]\big).
    \end{equation}
    
    \item Removing the cut-vertices $v_2$, and $v_6$ we get the following part,\begin{equation}
    (\lambda-5)(\lambda+4)\phi[1,3]\phi[4,5]\phi[7].
    \end{equation}
    
    Adding (a),(b),(c) and (d) give $\phi(M_1).$ 
    \end{enumerate}

    Now, we calculate the characteristic polynomial of $M_2$. Parts of the characteristic polynomial is listed below in terms of $\phi$-summands of the digraph $G(M_2)$.
    \begin{enumerate}
    \item 
    Without removing any cut-vertex we get the following part of $\phi(M_2)$.
    \begin{equation*}
    \begin{split}
    \phi[1,2,3]\phi[4,5,6]\phi[7]\phi[8]+\phi[1,2,3]\phi[4,5]\phi[6,7]\phi[8]+\phi[1,3]\phi[2,4,5,6]\phi[7]\phi[8]+\phi[1,3]\phi[2,4,5]\phi[6,7]\phi[8]\\ +\phi[1,2,3]\phi[4,5]\phi[7]\phi[6,8]+\phi[1,3]\phi[2,4,5]\phi[7]\phi[4,8].
    \end{split}
    \end{equation*}
    
    \item 
    The loop-weight of $v_2$ is $(5-\lambda)$ and its cut-index is 2. Removing the cut-vertex $v_2$ we get the following part, 
    \begin{equation*}
    \begin{split}
    (\lambda-5)\phi[1,3]\phi[4,5,6]\phi[7]\phi[8]+(\lambda-5)\phi[1,3]\phi[4,5]\phi[6,7]\phi[8]
    \\+(\lambda-5)\phi[1,3]\phi[4,5]\phi[7]\phi[6,8].
    \end{split}
    \end{equation*}
    
    \item The loop-weight of $v_6$ is $(-4-\lambda)$ and its cut-index is 3. Removing the cut-vertex $v_6$ we get the following part, \begin{equation}
    2(\lambda+4)\phi[1,2,3]\phi[4,5]\phi[7]\phi[8]+2(\lambda+4)\phi[1,3]\phi[2,4,5]\phi[7]\phi[8].
    \end{equation}
    
    \item Removing the cut-vertices $v_2$, and $v_6$ we get the following part,\begin{equation}
    2(\lambda-5)(\lambda+4)\phi[1,3]\phi[4,5]\phi[7]\phi[8].
    \end{equation}
    
    Adding (a),(b),(c) and (d) give $\phi(M_2).$ 
    \end{enumerate}
    \end{enumerate}
    \end{example}

    \section{Determinant of digraphs}\label{detdi}
    Note that, the determinant of a matrix can be calculated from its characteristic polynomial by setting $\lambda = 0$. Thus, theorem \ref{matrixcp} is applicable for calculating determinant of a matrix corresponding to the digraph $G_k$ by setting $\lambda=0$, and all $\phi$-summands replaced by $\det$-summands [see definition \ref{def1}]. 
    
    \begin{corollary} \label{deter}
    Let $G_k$ be a digraph having no loops on its cut-vertices, then determinant of $G_k$ is given by sum of \emph{$\det$}-summands of all possible $\mathcal{B}$-partition of $G_k$. 
    \end{corollary}
    \begin{proof}
    As all the cut-vertices has no loops, all $\det$-summands of resulting induced diagraph after removing any cut-vertices get multiplied by zero. Thus, only those $\det$-summands contributes which are corresponding to $q=0$ . Hence, the result follows from theorem \ref{matrixcp}. 
    \end{proof}
    
    A digraph $G$ whose determinant is zero is called a singular digraph. The next corollary provides a number of singular simple graphs. For definition of complete graph, cycle graph, tree, and forest we refer \cite{west2001introduction, harary1969graph}.
    \begin{corollary}
    A simple graph $G$ is singular, if any of these followings holds:
    \begin{enumerate}
    \item 
    There is a pendant block $C_n$ with cut-vertex $v$ of $G$. Here, $C_n$ is a cyclic graph with $n=4r$, $r$ is a positive integer.
    \item 
    There are two pendant blocks $C_n$, and $C_m$ sharing a cut-vertex $v$ of $G$. Here, $C_n$, and $C_m$ are cyclic graphs, where $n$ and $m$ are even positive integers. 
    \item
    There is a singular tree with positive even number of vertices. It contains a cut-vertex $v$ of $G$.
    \item 
    There are two trees with $n_1$ and $n_2$ vertices which share a common cut-vertex $v$ of $G$. Here, both $n_1$,and $n_2$ are either even or both of them odd positive integers.
    \end{enumerate}
    \begin{proof}
    Below, we prove all these conditions separately.
    \begin{enumerate}
    \item 
    There are two types of $\det$-summands of $G$: one consists of determinant of $C_n$, and another consists of determinant of $C_n\setminus v$, that is a path of odd length. Now, determinants of $C_n$ and $C_n\setminus v$ are zero \cite{germina2010signed}.  Hence, the first condition follows.
    
    \item 
    In $\det$-summand of $G$ each term will either have determinants of $C_n$ and $C_m\setminus v$ or  determinants of $C_n\setminus v$ and $C_m$. Note that, $C_m\setminus v$, and $C_n\setminus v$ are paths of odd length which are singular \cite{germina2010signed}. Hence, result follows.
     
    \item 
    Applying the Lemma (\ref{recurrsubp}), we expand the determinant of $G$, with respect to the cut-vertex $v$. In this expression, every term has either determinant of tree or determinant of a forest with an odd number of vertices. A forest with an odd number of vertices is singular \cite{bapat2010graphs}. Hence the result follows.  
    \item In the determinant expansion of $G$ with respect to cut-vertex $v$ each term will have at least one term having the determinant of forest having an odd number of vertices. Hence, the result follows.

    \end{enumerate}
    \end{proof}

    \end{corollary}
    \subsection{Determinant of Block Graph} \label{key}
    Complete graph on $n$ vertices is denoted by $K_n$ and its determinant is $(-1)^{n-1}(n-1)$. For a simple graph $G$, when each of its blocks is a complete graph then $G$ is called block graph \cite{bapat2014adjacency}. The determinant of a block graph can be calculated from the theorem below, quoted from \cite{bapat2014adjacency}. Interested readers may see a recurrence relation for block graph in \cite{247313}. Here, we present a simplified form of this theorem from the Corollary \ref{deter}. 
    \begin{theorem}
    Let $G_k$ be a block graph with $n$ vertices. Let $B_1, B_2,\hdots, B_k$ be its blocks of size $b_1, b_2, \hdots, b_k,$ respectively. Let $A$ be the adjacency matrix of $G$. Then $$\det (A)=(-1)^{n-k}\sum\prod_{i=1}^{k}(\alpha_i-1),$$ 
    where the summation is over all $k$-tuples $(\alpha_1, \alpha_2, \hdots,\alpha_k)$ of non negative integers satisfying the
    
    following conditions:
    \begin{enumerate}
    \item $\sum_{i=1}^{k} \alpha_i=n$;
    \item for any nonempty set $S\subseteq \left\{1,2,\hdots,k \right\}$ $$\sum_{i\in S}\alpha_i\le|V(G_S)|,$$ where $G_S$ denote the subgraph of $G$ induced by the blocks $B_i$, $i\in S$.
    \end{enumerate} 
    \end{theorem}
    
    \begin{proof}
    From the Corollary \ref{deter} determinant of block graph is equal to 
    $$\sum \prod_{i=1}^{k}\det(\hat{B_i}),$$ 
    where $\hat{B_i}$ is subgraph of $B_i$ and summation is over all $\mathcal{B}$-partition of $G_k$.
    
    Contribution of the $\det$-summand of a $\mathcal{B}$-partition with induced subgraphs $\hat{B_1}, \hat{B_2},\hdots, \hat{B_k}$ of sizes $\beta_1, \beta_2,\hdots, \beta_k$, respectively, is following, $$\prod_{i=1}^{k}(-1)^{\beta_i-1}(\beta_i-1).$$   
    As $(\hat{B_1}, \hat{B_2}, \hdots,\hat{B_k})$ are vertex disjoint induced subgraphs which partition $G_k$, $\sum_{i=1}^{k}\beta_i=n$. Thus, contribution of the above tuple can be written as, 
    $$(-1)^{n-k}\prod_{i=1}^{k}(\beta_i-1).$$
    
    Also, for any nonempty set $S\subset \left\{1,2,\hdots,k \right\}$,
    $$\sum_{i\in S}|V(\hat{B_i})|=\sum_{i\in S}\beta_i\le|V(G_S)|,$$ 
    where $G_S$ denote the subgraph of $G$ induced by the blocks $B_i$, $i\in S$. Thus, $k$-tuples  $(\beta_1, \beta_2,\hdots,\beta_k)$ resulted from $\mathcal{B}$-partitions of $G_k$ satisfy both the conditions of theorem.     
    
    Conversely, consider a $k$-tuple $(\alpha_1, \alpha_2, \hdots, \alpha_k)$ satisfying both the condition of theorem. We will prove by induction that each such $k$-tuple corresponds to a unique $\mathcal{B}$-partition of $G_k.$ 
    
    If $G_1$ has only one block $B_1$ of size $b_1$, then $G_1$ is a complete graph, $K_{b_1}$. The only possible choice for $1$-tuple is $\alpha_1 = b_1$. Clearly, $\alpha_1$ corresponds to a $\mathcal{B}$-partition which consists of $K_{b_1}$ only. Let $G_2$ has block $B_1$, and $ B_2$ of sizes $b_1$, and $b_2$, respectively.  The possible $2$-tuples are ($\alpha_1=b_1, \ \alpha_2=b_2-1$), and ($\alpha_1=b_1-1, \ \alpha_2=b_2$). Therefore, both the $2$-tuple induce possible two $\mathcal{B}$-partitions in $G_2.$ One $\mathcal{B}$-partition consists of induced subgraphs $K_{b_1}, K_{b_2-1}$. Another $\mathcal{B}$-partition consists of induced subgraphs $K_{b_1-1}, K_{b_2}$.
    
    Now we discuss the proof for $G_3$, which will clarify the reasoning for the general case. In $G_3$ block $B_3$ can occur in two ways.
    \begin{enumerate}
    \item 
    Let $B_3$ be added to a non cut-vertex of $G_2$. Without loss of generality, let $B_3$ is attached to a non-cut-vertex of $B_2$. Choices for $3$-tuple $(\alpha_1, \alpha_2, \alpha_3)$ are following: 
    \begin{enumerate}
    \item $\alpha_1=b_1,\ \alpha_2=b_2-1,\ \alpha_3=b_3-1$;
    \item $\alpha_1=b_1,\ \alpha_2=b_2-2,\ \alpha_3=b_3$;
    \item $\alpha_1=b_1-1,\ \alpha_2=b_2,\ \alpha_3=b_3-1$;
    \item $\alpha_1=b_1-1,\ \alpha_2=b_2-1,\ \alpha_3=b_3$.
    \end{enumerate}
    Note that, in this case, each 2-tuple of $G_2$ give rise to two 3-tuple in $G_3$ where $\alpha_1$ is unchanged. Clearly, all the tuples in $G_3$ can induce its all possible $\mathcal{B}$-partitions.  
    
    \item Let $B_3$ be added to cut-vertex $v$ of $G_2.$ Choices for $3$-tuple $(\alpha_1, \alpha_2, \alpha_3)$ are following: 
    \begin{enumerate}
    \item $\alpha_1=b_1,\ \alpha_2=b_2-1,\ \alpha_3=b_3-1$;
    \item $\alpha_1=b_1-1,\ \alpha_2=b_2,\ \alpha_3=b_3-1$;
    \item $\alpha_1=b_1-1,\ \alpha_2=b_2-1,\ \alpha_3=b_3$.
    \end{enumerate} 
    Here, each 2-tuple of $G_2$ give rise to a 3-tuple of $G_3$ where $\alpha_1, \alpha_2$ are unchanged and $\alpha_3=b_3-1$. Beside these there is one more 3-tuple where $\alpha_1=b_1-1,\ \alpha_2=b_2-1, \alpha_3=b_3$. Clearly, all the tuples in $G_3$ can induce its all possible $\mathcal{B}$-partitions.   
    \end{enumerate}
    
    Now let us assume that all possible $m$-tuples $(\alpha_1,\alpha_2,\hdots,\alpha_m)$ in $G_m$ can induce all possible $\mathcal{B}$-partitions in it. We need to proof that all possible $(m+1)$-tuples $(\alpha_1,\alpha_2,\hdots,\alpha_m, \alpha_{m+1})$ in $G_{m+1}$ can induce its all possible $\mathcal{B}$-partitions in it. In $G_{m+1}$ block $B_{m+1}$ can occur in two ways.
    \begin{enumerate}
    \item Let $B_{m+1}$ be added to non cut-vertex of $G_m.$ 
    Each $m$-tuple $(\alpha_1,\alpha_2,\hdots,\alpha_m)$ of $G_k$ give rise to two $(m+1)$-tuple of $G_{m+1}$ where, $\alpha_1, \alpha_2,\hdots,\alpha_{m-1}$ are unchanged. In one such tuple $\alpha_m$ is also unchanged and $\alpha_{m+1}=b_{m+1}-1.$ In other tuple $\alpha_{m}$ is one less than the value it had earlier and $\alpha_{m+1}=b_{m+1}.$ Thus, $(m+1)$-tuples in $G_{m+1}$ can induce all its $\mathcal{B}$-partitions in $G_{k+1}.$ 
    
    \item Let $B_{k+1}$ be added to a cut-vertex $v$ of $G_k.$ 
     Each $k$-tuple of $G_k$ give rise to one $(k+1)$-tuple of $G_{k+1}$ where $\alpha_{k+1}=b_{k+1}-1.$
    Beside these there are also $(k+1)$-tuples where $\alpha_{k+1}=b_{k+1}$, along with $k$-tuples of $(G_k\setminus v)$. Clearly, all the tuples in $G_{k+1}$ can induce its $\mathcal{B}$-partitions. 
    \end{enumerate}
    Hence, there is one to one correspondence between partitions and $k$-tuples. 
    \end{proof}

    \section{Digraph with no cut-vertex}\label{future}
    A digraph $G$ with no cut-vertex is a block itself. Determinant of $G$ can be computed in terms of determinants of its subdigraphs which may contain blocks. Hence, we may further simplify calculation of determinant of $G$. Let $A_1$ be a matrix of order $(n-1)$, as well as $\mathbf{b}$, and $\mathbf{c}$ are a column vector, and a row vector of order $n-1$, respectively. With suitable rebelling of vertices in digraph $G(A)$ we may write,
    $$A=\begin{bmatrix}
    A_1 & \mathbf{b}\\ \mathbf{c}& d
    \end{bmatrix}.$$
    If $A_1$ is invertible then from Schur' complement for determinant \cite{bapat2014adjacency},
    $$\det (A)=\det (A_1) \det(d-\mathbf{c}A^{-1}_1\mathbf{b}).$$
    
    If $A_1$ is not invertible, and $d$ is non-zero then, $$\det (A)=d \times \det(A_1-\mathbf{b}\frac{1}{d}\mathbf{c}).$$
    
    If $A_1$ is not invertible, and $d$ is zero then, $$\det (A)= \det \begin{bmatrix}
    A_1 & \mathbf{b}\\ \mathbf{c}& 0
    \end{bmatrix}=\det \begin{bmatrix}
    A_1 & \mathbf{b}\\ \mathbf{c}& 1
    \end{bmatrix}-1\times \det( A_1)=\det(A_1-\mathbf{b}\mathbf{c}).$$
    
    In all the above cases, we observe that det($A$) is calculated with the corresponding digraph $G(A)$, in terms of a lower order matrices. If any of these digraphs of lower order matrices have blocks, we can apply results of previous sections. Thus, $\det (A)$ can be computed recursively. For an efficient recursive method, we should prefer a vertex $v$, such that, the number of blocks in $G\setminus v$ is maximum. Note that, in this context, the digraph $G\setminus v$ corresponds to $A_1$. Also, $\mathbf{b},d$ and $\mathbf{c},d$ form column vector, and row vector corresponding the to vertex $v$. This lead us to propose an open problem \ref{open} in the conclusion.
    
           \begin{figure}[h!]
                        \begin{center}
                          \includegraphics[width=.4\linewidth]{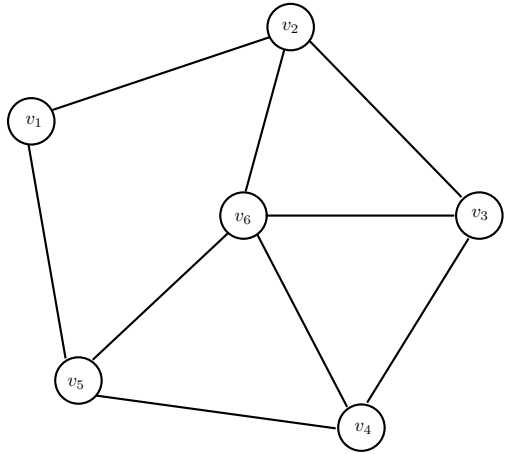}
                                  \end{center}
                        \caption{A digraph of order 6 having no cut-vertex}
                        \label{nos}
            \end{figure}

    \section{Conclusion} \label{con}
    In this article we found characteristic, and permanent polynomials of digraphs in terms of $\phi$-summands, and $\psi$-summands of corresponding $\mathcal{B}$-partition. Our result gave a simple proof for determinant of block graph. We conclude the paper with some open problems:
    \begin{enumerate}
    \item \label{open}
        Let $G$ be a digraph with no cut-vertex. Also, assume $b(G)$ denotes the number of blocks in $G$. Then find out a vertex $v$, such that, $b(G\setminus v)\geq b(G\setminus u)$ for all $u$ in $G$. Degree of a vertex $v$ is number of vertices adjacent to it. In most of the digraphs having no cut-vertex, removing a vertex having maximum degree introduce maximum number of blocks in the  resultant digraph. This observation is true for all digraphs of order less than $6$,\cite{bicon} but there are counterexamples in higher order. For example in digraph in figure \ref{nos}: vertex $v_2$ has degree 3, and vertex $v_6$ has degree 4. But, removing $v_2$ results in 2 blocks while removing $v_6$ results in 1 block. 

    \item
        Matrices are classified into a number of families possessing different properties. A rigorous study can be done on their characteristic polynomials with this combinatorial methods. It will provide us a combinatorial description of those particular properties.
    \item
        Finding singular graphs is an important problem due to their applications in quantum chemistry, and Huckel's molucular orbital theory. This article provides a new direction of research to find out singular graphs.
    \item
        The Schur's complement method is not applicable for finding out permanent of a matrix. It makes this problem complicated. Here we have calculated permanent for graphs with a number of cut vertiecs. It can be extended for further investigations for those diagraphs having no cut vertex.
    \item
        The Laplacian matrices of graphs provide a pictorial description of quantum states in the quantum mechanics and information theory. We may do the similar study where the Laplacian matrix is involved instead of graph adjacency matrix. It will be helpful to represent quantum states with weighted digraphs. 
    \end{enumerate}

    \bibliographystyle{plain}
    \bibliography{SB}

    \section{Appendix}
       \begin{figure}[h!]
            \begin{center}
    \includegraphics[width=1\linewidth]{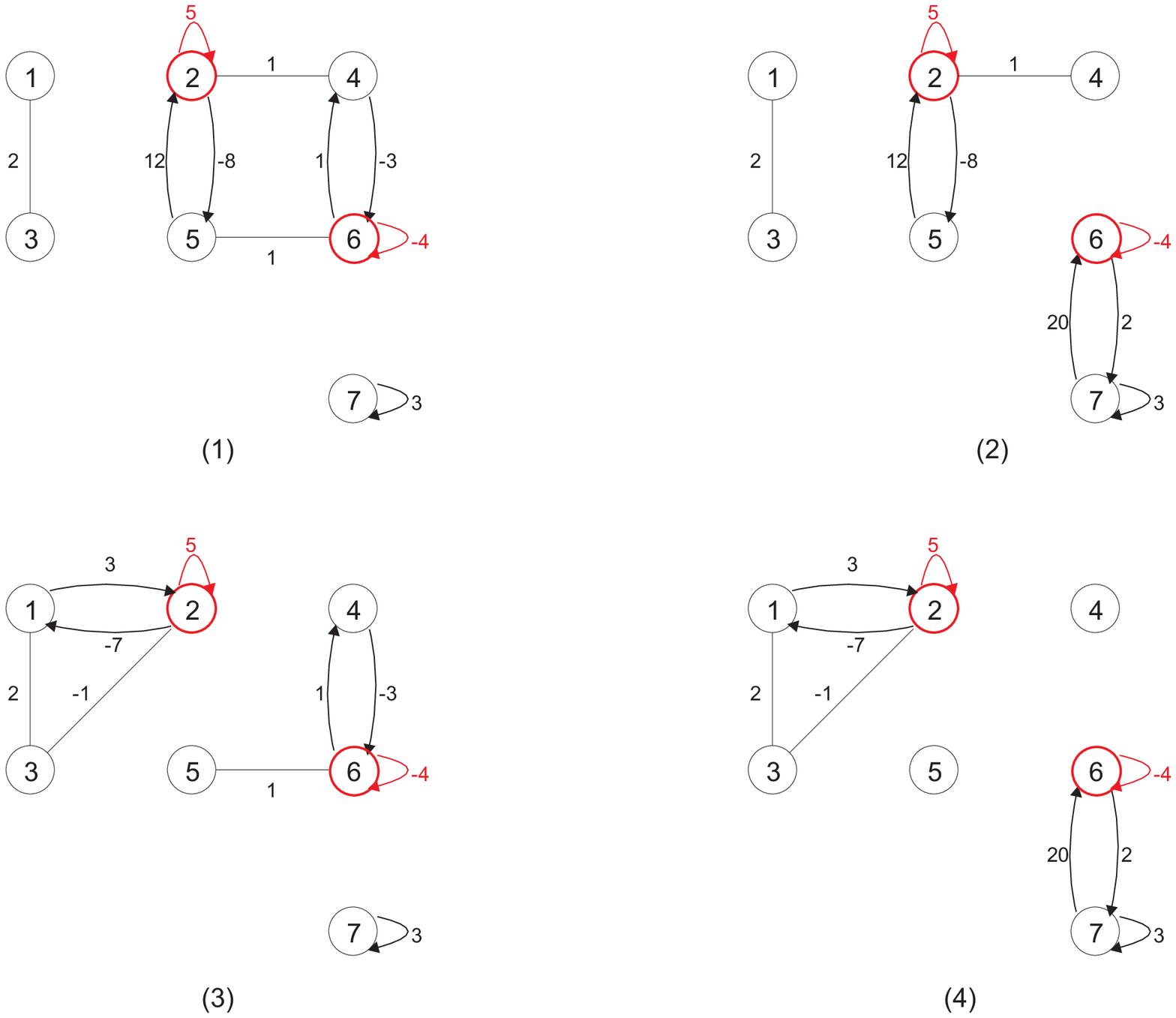}
            \end{center}
              
            \caption{$\mathcal{B}$-partitions of digraph of matrix $M_1$}
            \label{allbm1}
            \end{figure}    
        
            \begin{figure}
                      \begin{center}
                      \includegraphics[width=1.1\linewidth]{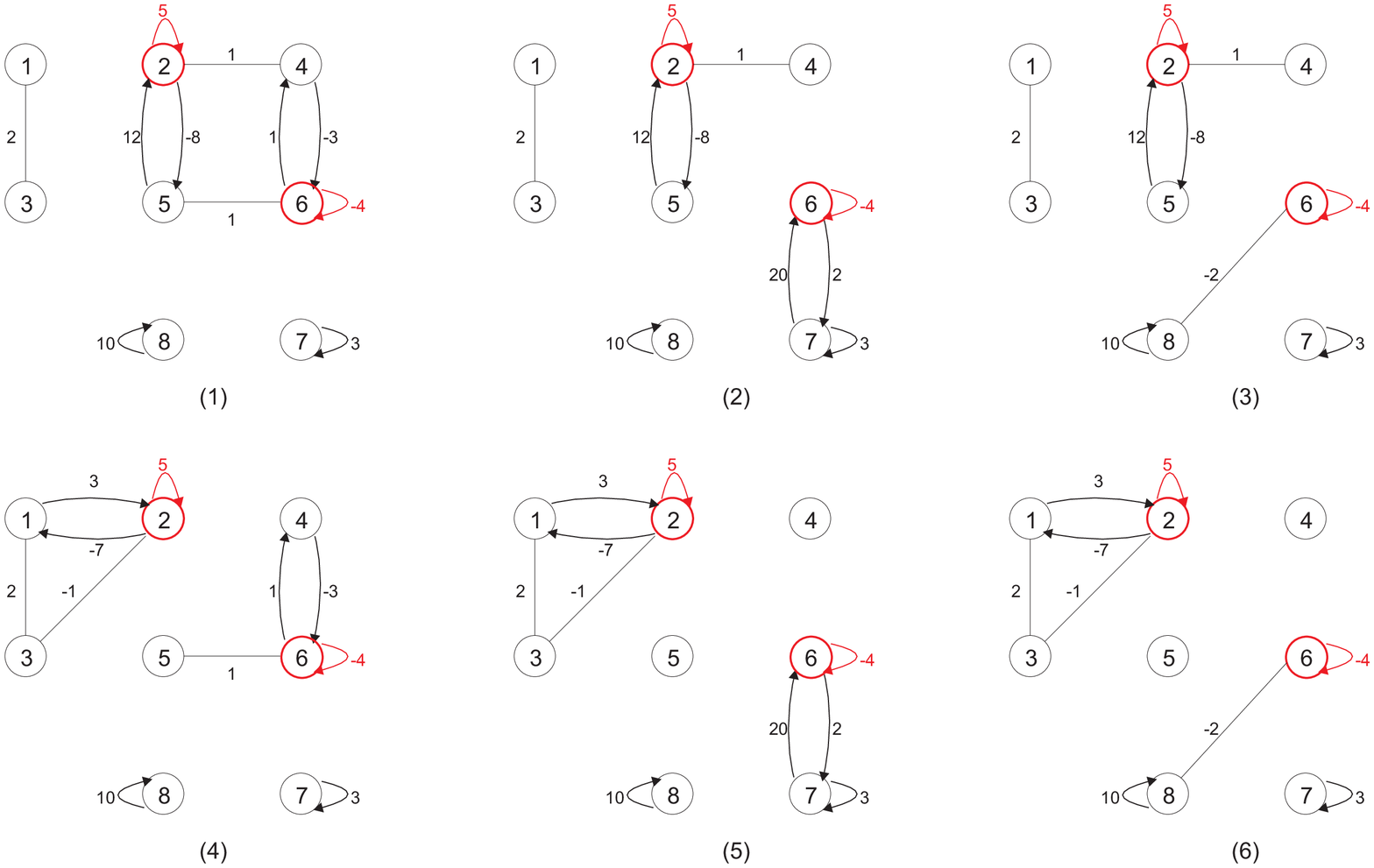}
                              \end{center}
                    \caption{$\mathcal{B}$-partitions of digraph of matrix $M_2$}
                    \label{allbm2}
        \end{figure}

\end{document}